\documentclass[]{TrucLeStyle}

\begin{document}

\title{{\itshape Intrinsic Prices Of Risk}}

\author{TRUC LE}

\affil{}


\maketitle

\begin{abstract}
    We review the nature of some well-known phenomena such as volatility smiles, convexity adjustments and parallel derivative markets. We propose that the market is incomplete and postulate the existence of an intrinsic risk in every contingent claim as a basis for understanding these phenomena. In a continuous time framework, we bring together the notion of intrinsic risk and the theory of change of measures to derive a probability measure, namely risk-subjective measure, for evaluating contingent claims. This paper is a modest attempt to prove that measure of intrinsic risk is a crucial ingredient for explaining these phenomena, and in consequence proposes a new approach to pricing and hedging financial derivatives. By adapting theoretical knowledge to practical applications, we show that our approach is consistent and robust, compared with the standard risk-neutral approach.

\end{abstract}

\section{Introduction}

This section has two purposes. Firstly, we review some well-known phenomena in order to motivate subsequent developments. After that, we provide a background of the phenomena with some notation, terminology and notions.

\subsection{Phenomena}

\textit{\textbf{Volatility smiles.}} In a nutshell, vanilla options with different maturities and strikes have different volatilities implied by the well-known formula of \cite{bs73}. Implied volatility is quoted as the market expectation about the average future volatility of the underlying asset over the remaining life of the option. Thus compared to historical volatility it is the forward looking approach.

For many years, practitioners and academics have tried to analyse the volatility smile phenomenon and understand its implications for derivatives pricing and risk management. In \cite{cr76}, their link between the real-world and risk-neutral processes of the underlying would be complete by non-traded sources of risk. \cite{scott87} found that the dynamics of the risk premium, when volatility is stochastic, is not a traded security. A number of models and extensions of, or alternatives to, the Black-Scholes model, have been proposed in the literature: the local volatility models of \cite{dupire94}, \cite{dk94}; a jump-diffusion model of \cite{mer76}; stochastic volatility models of \cite{hw88}, \cite{hes93} and others; mixed stochastic jump-diffusion models of \cite{bates96} and others; universal volatility models of \cite{dupire96}, \cite{jpm99}, \cite{bn00}, \cite{bla01} and others; regime switching models, etc.

From a hedging perspective, traders who use the Black-Scholes model must continuously change the volatility assumption in order to match market prices. Their hedge ratios change accordingly in an uncontrolled way: the models listed above bring some order into this chaos. In the course of time, the general consensus, as advocated by practitioners and academics, is to \emph{choose} a model that produces hedging strategies for both vanilla and exotic options resulting in profit and loss distributions that are sharply peaked at zero. We argue that such a model, if recovered (or implied) from option prices, by no means nearly explains this phenomenon, but is a means only to describe the implied volatility surface.

\textit{\textbf{Convexity adjustments.}} One of many well-known adjustments is the convexity adjustment - the implied yield of a futures and the equivalent forward rate agreement contracts are different. This phenomenon implies that market participants need to be paid more (or less) premium.

The common approach, as used by most practitioners and academics, is to adjust futures quotes such that they can be used as forward rates. Naturally, this approach depends on an model that is used for this purpose. For the extended Vasicek known as \cite{hw90} and \cite{cir85} model, explicit formulae can be derived. The situation is different for models whose continuous description gives the short rate a log-normal distribution such as the \cite{bdt90} and \cite{bk91} models: for these, in their analytical form of continuous evolution, futures prices can be shown to be positively infinite \cite{hjm92} and \cite{ss94}. In subsequent developments, we shall offer a different approach to this phenomenon.

\textit{\textbf{Parallel derivative markets.}}
In an economic system, a financial market consists of a risk-free money account, \emph{primary} and \emph{parallel} markets. Examples of primary markets are stocks and bonds, and examples of parallel markets are derivatives such as forward, futures, vanilla options, credits which are derived from the same primary asset. Market makers can trade and make prices for derivatives in a parallel market without references to another.

\subsection{Background}

The framework is as follows: a complete probability space $(\Omega,\mathcal{F},\mathbb{P})$ with a filtration $\mathcal{F}= \mathcal{F}(t)$ satisfying the usual conditions of right-continuity and completeness. $T\in\mathbb{R}$ denotes a fixed and finite time horizon; furthermore, we assume that $\mathcal{F}(0)$ is trivial and that $\mathcal{F}(T) = \mathcal{F}$. Let $X = X(t)$ be a continuous semimartingale representing the price process of a risky asset.

The absence of arbitrage opportunities implies the existence of an probability measure $\mathbb{Q}$ equivalent to the probability measure $\mathbb{P}$ (the real world probability), such that $X$ is a $\mathbb{Q}$-martingale. Denote by $\mathcal{Q}$ the set of coexistent equivalent measures $\mathbb{Q}$. A financial market is considered such that $\mathcal{Q} \neq \emptyset$. Uniqueness of the equivalent probability measure $\mathbb{Q}$ implies the market is \emph{complete}. The fundamental theorem of asset pricing establishes the relationship between the absence of arbitrage opportunities and the existence of an equivalent martingale measure and in a basic framework was proved by \cite{hk79}, \cite{hp81,hp83}. The modern version of this theorem, established by \cite{ds04}, states that the absence of arbitrage opportunities is ``essentially" equivalent to the existence of an equivalent martingale measure under which the discounted (primary asset) price process is a martingale.

For simplicity, we consider only one horizon of uncertainty $[0,T]$. A \emph{contingent claim}, or a \emph{derivative}, $H = H(\omega)$ is a payoff at time $T$, contingent on the scenario $\omega\in\Omega$. The derivative has the special form $H = h(X(T))$ for some function $h$. Here, $X$ is referred to as the primary (or the `underlying'). More generally, $H$ depends on the whole evolution of $X$ up to time $T$ and is a random variable
\begin{equation}
    H \in \mathcal{L}^2(\Omega,\mathcal{F},\mathbb{P}).
\end{equation}

In financial terms, every contingent claim can be replicated by means of a \emph{trading strategy} (or interchangeably known as \emph{hedging strategy} or a \emph{replication portfolio}) which is a portfolio consisting of the primary asset $X$ and a risk-free money account $D = D(t)$. Let $\alpha = \alpha(t)$ and $\beta =\beta(t)$ be a predictable process and an adapted process, respectively. $\alpha(t)$ and $\beta(t)$ are the amounts of asset and money account, respectively, held at time $t$. In this section, for ease of exposition, we assume that $D(t) = 1$ for all $0\le t\le T$. The \emph{value} of the portfolio at time $t$ is given by
\begin{equation}\label{portvalue}
    V(t) = \alpha(t)X(t) + \beta(t)D(t)
\end{equation}
for $0\le t\le T$. It can be shown that the trading strategy $(\alpha,\beta)$ is \emph{admissible} such that the value process $V = V(t)$ is square-integrable and have right-continuous paths and is defined by
\begin{equation}
    V(t) := V_0 + \int_0^t\alpha(s)dX(s)
\end{equation}
for $0\le t\le T$. For $\mathbb{Q}$-almost surely, every contingent claim $H$ is attainable and admits the following representation
\begin{equation}
    V(T) = H = V_0 + \int_0^T\alpha(s)dX(s),
\end{equation}
where $V_0 = E_{\mathbb{Q}}[H]$. Moreover, the strategy is \emph{self-financing}, that is the cost of the portfolio (also known as derivative price) is a constant $V_0$
\begin{equation}\label{perfecthedge}
    V(t) - \int_0^t\alpha(s)dX(s) = V_0.
\end{equation}
The constant value $V_0$ represents a perfect replication or a \emph{perfect hedge}.

Thus far, we have presented the well-known mathematical construction of a hedging strategy in a complete market where every contingent claim is attainable. In a complete market, derivative prices are unique - no arbitrage opportunities exist. Derivatives cannot be valuated in a parallel market at any price other than $V_0$.

From financial and economic point of view, the phenomena imply that the market is \emph{incomplete}, arbitrage opportunities exist and may not be at all eliminated. A derivative can be valued at different prices and hedged by mutually exclusively trading in risky assets (or derivatives) in parallel markets where market makers engage in market activities: investments, speculative trading, hedging, arbitrage and risk management. In addition, market makers expose themselves to market conditions such as liquidity, see for instance \cite{bac93}. We argue that exposure to the variability of market activities, market conditions and generally to uncertain future events constitutes a basis of arbitrage opportunity which we shall call \emph{intrinsic risk}.

In general, market incompleteness is a principle under which every contingent claim bears intrinsic risks. Let us postulate an assumption as a basis for subsequent reasonings and discussions.\\

\textbf{Assumption.} \textit{The market is incomplete and there exist intrinsic risks embedded in every contingent claim.}\\

While the assumption is theoretical, it is rather realistically a proposition with the phenomena as proof.

In a mathematical context, let $\Pi$ be the set of all intrinsic risks, that is the set of all real valued functions on $\Omega$. Denote by $G(\pi)$ the measure to the intrinsic risk $\pi = \pi(\omega)$ on the scenario $\omega\in\Omega$. As a measure of intrinsic risk, $G$ is a mapping from $\Pi$ into $\mathbb{R}$. As a basic object of our study, $G$ shall therefore be the random variable on the set of states of nature at a future date $T$. Generally, $G$ depends on the evolution of the primary asset up to time $T$ and may also depend on the contingent claim:
\begin{equation}
    G^H \in \mathcal{L}^2(\Omega,\mathcal{F},\mathbb{P}).
\end{equation}
The superscript indicates the dependence of a particular contingent claim $H$. This leads to a new representation of $H$
\begin{equation}\label{HGrepresenttation}
    H = V_0 + \int_0^T\alpha(s)dX(s) + G^H.
\end{equation}
We now introduce the Kunita-Watanabe decomposition
\begin{equation}
    G^H = G_0 + \int_0^T\alpha^H(s)dX(s) + N(T)
\end{equation}
where $N = N(t)$ is a square-integrable martingale orthogonal to $X$. Thus, we have
\begin{equation}\label{Hrep}
    H = V_0^* + \int_0^T\alpha^*(s)dX(s) + N(T),
\end{equation}
where $V_0^* = V_0 + G_0$ and $\alpha^* = \alpha + \alpha^H$. This representation of $H$ have been extensively dealt with, see for example \cite{fs91}. By incompleteness, the derivative value $V^*_0$ represents a perfect hedge, which manifests an initial intrinsic value of risk $G_0$. In relation to the hedging strategy (\ref{HGrepresenttation}), the measure of intrinsic risk shall be considered as the value of all possible future capital which, required to control the risk incurred by the market maker (such as hedger) and invested in the primary asset, makes not only the contingent claim acceptable, but its valuation fair.

From a mathematical point of view, market incompleteness implies that there exists in the set $\mathcal{Q}$ an equivalent measure, not necessarily a martingale and/or unique measure, that is assigned to a parallel market. Thus, intrinsic risk may depend on the derivative and is not necessarily unique, as such its measure takes many forms some of which we shall consider for applications. In the remaining of this paper, we shall not discuss further on the abstract representations (\ref{HGrepresenttation}) and (\ref{Hrep}), but present them in a more descriptive (down to earth) framework - the continuous time framework.

\section{Market, Portfolio, Absence of Arbitrage and Intrinsic Price of Risk}

In this section we propose a continuous time financial market consisting of a primary price process $X$ and a risk-free money account $D$. We shall define a measure of intrinsic risk and show that perfect hedging strategies can be constructed. We also show that the existence of intrinsic risk provides an internal consistency in pricing and hedging a contingent claim.

Let $B = B(t)$ be a Brownian motion on the complete probability space $(\Omega,\mathcal{F},\mathbb{P})$. The underlying price process of $X$ satisfies the SDE
\begin{equation}\label{XPsde}
    dX(t) = \mu(t)X(t)dt + \sigma(t)X(t)dB(t),
\end{equation}
where $\mu = \mu(t)$ and $\sigma = \sigma(t)$ are Lipschitz continuous functions so that a solution exists. $\mu$ and $\sigma$ can be functions of $X$. The price process of $D$ is given by
\begin{equation}
    dD(t) = \nu(t)D(t)dt,
\end{equation}
where $\nu = \nu(t)$ is a Lipschitz continuous function.

\smallskip
We expand the portfolio value process (\ref{portvalue}) as follows:
\begin{eqnarray}\label{Qhedgestrat}
    dV(t) &=& \alpha(t)dX(t) + \nu(t)\beta(t)D(t)dt \\
        &=& \alpha(t)\mu(t)X(t)dt + \alpha(t)\sigma(t)X(t)dB(t) + \nu(t)\left(V(t) - \alpha(t)X(t)\right)dt \nonumber\\
        &=& \nu(t)V(t)dt + \alpha(t)(\mu(t) - \nu(t))X(t)dt + \alpha(t)\sigma(t)X(t)dB(t)\nonumber \\
        &=& \nu(t)V(t)dt + \alpha(t)\sigma(t)X(t) \left[ \frac{\mu(t) - \nu(t)}{\sigma(t)}dt + dB(t)\right]\nonumber \\
        &=& \nu(t)V(t)dt + \alpha(t)\sigma(t)X(t) dW(t),\nonumber
\end{eqnarray}
where $W = W(t)$ is a $\mathbb{Q}$-Brownian motion and is defined by
\begin{equation}
    dW(t) = \lambda(t)dt + dB(t)
\end{equation}
and
\begin{equation*}
    \lambda(t) = \frac{\mu(t) - \nu(t)}{\sigma(t)}.
\end{equation*}
Here, $\mathbb{Q}$ is some martingale measure. Indeed, the theory of the Girsanov change of measure, see for example \cite{ks98}, shows that there exists such a martingale measure $\mathbb{Q}$ equivalent to $\mathbb{P}$ and which excludes arbitrage opportunities. More precisely, there exists a probability measure $\mathbb{Q} \ll \mathbb{P}$ such that
\begin{equation}
    \frac{d\mathbb{Q}}{d\mathbb{P}} \in \mathcal{L}^2(\Omega,\mathcal{F},\mathbb{P})
\end{equation}
and $X$ is a $\mathbb{Q}$-martingale. Such a martingale measure $\mathbb{Q}$ is determined by the right-continuous square-integrable martingale
\begin{equation*}
    \Lambda(t) = E_{\mathbb{P}}\left[ \left. \frac{d\mathbb{Q}}{d\mathbb{P}} \right| \mathcal{F}(t)\right]
\end{equation*}
for $0\le t\le T$. And explicitly
\begin{equation*}
    \Lambda(T) = \exp\left(-\int_0^T\lambda(t)dB(t) - \frac{1}{2}\int_0^T\lambda^2(t)dt\right)
\end{equation*}
and $\lambda$ satisfies Novikov's condition
\begin{equation*}
    E_{\mathbb{P}}\left[\exp\left( \frac{1}{2}\int_0^T\lambda^2(t)dt \right)\right] < \infty.
\end{equation*}
It is not hard to see that the price process $X$ under $\mathbb{Q}$ is given by
\begin{equation}\label{XQsde}
    dX(t) = \nu(t)X(t)dt + \sigma(t)X(t)dW(t).
\end{equation}

Note that the martingale measure $\mathbb{Q}$ and $\lambda$ are, \emph{if unique}, theoretically and practically well-known as the \emph{risk-neutral measure} and the \emph{market price of risk}, respectively. The risk-neutral valuation formula is given by
\begin{equation}\label{Qvalformula}
    V(t) = D(t)E_{\mathbb{Q}}\left[\left. \frac{1}{D(T)}H \right| \mathcal{F}(t) \right] = D(t)E_{\mathbb{Q}}\left[\left. \frac{1}{D(T)}h(X(T)) \right| \mathcal{F}(t) \right].
\end{equation}
The expectation is taken under the measure $\mathbb{Q}$.

It is important to note that in the risk-neutral world the essential theoretical assumptions are: (1) the \emph{true} price process (\ref{XPsde}) is correctly specified and (2) prices of derivatives $H$ are drawn from this price process, that is derivative prices are uniquely determined by formula (\ref{Qvalformula}). These assumptions, if not violated, lead to a complete market and the trading strategy (\ref{Qhedgestrat}) and the measure $\mathbb{Q}$ are unique. However, in practice as we argued earlier, these assumptions are strongly violated; as a result market completeness and uniqueness of derivative prices are no longer valid. That is $\mathbb{Q}$ is no longer risk-neutral, but only an equivalent measure in the set $\mathcal{Q}$.

 We now consider the representation (\ref{HGrepresenttation}) in a continuous time framework, the measure of intrinsic risk can be defined, without loss of generality, in terms of changes in values in a future time interval $[t, t+dt]$ as follows.

\textbf{Definition.} \textit{The measure of intrinsic risk in a time interval $dt$ is defined by $dG(t,T) = \zeta(t,T)X(t)dt$, where $\zeta = \zeta(t,T)$ is a continuous adapted process representing a rate of intrinsic risk.}

As was represented earlier in (\ref{HGrepresenttation}), the evolution of a trading strategy shall be adaptable to adjust for the measure of intrinsic risk which can be considered an additional/less capital required in a time interval $dt$, that is
\begin{eqnarray}\label{Oportmod}
    dV(t) &=& \alpha(t)(dX(t) + dG(t)) + \nu(t)\beta(t)D(t)dt \\
        &=& \alpha(t)(\mu(t) + \zeta(t,T))X(t)dt + \alpha(t)\sigma(t)X(t)dB(t) + \nu(t)\left(V(t) - \alpha(t)X(t)\right)dt \nonumber\\
        &=& \nu(t)V(t)dt + \alpha(t)(\mu(t) + \zeta(t,T) - \nu(t))X(t)dt + \alpha(t)\sigma(t)X(t)dB(t)\nonumber \\
        &=& \nu(t)V(t)dt + \alpha(t)\sigma(t)X(t) \left[ \frac{\mu(t) + \zeta(t,T) - \nu(t)}{\sigma(t)}dt + dB(t)\right]\nonumber \\
        &=& \nu(t)V(t)dt + \alpha(t)\sigma(t)X(t)dZ(t),\nonumber
\end{eqnarray}
where $Z = Z(t)$ is a $\mathbb{S}$-Brownian motion and is given by
\begin{equation}\label{ZObrownian}
    dZ(t) = \frac{\mu(t) + \zeta(t,T) - \nu(t)}{\sigma(t)}dt + dB(t) = \frac{\zeta(t,T)}{\sigma(t)}dt + dW(t)
\end{equation}
and $\mathbb{S}$ is a measure equivalent to $\mathbb{P}$. Thus, $\mathbb{S} \in \mathcal{Q}$. Analogously, $\zeta/\sigma$ is defined as an \emph{intrinsic price of risk}.

Under $\mathbb{S}$ measure, the price process of $X$ under $\mathbb{S}$ is given by
\begin{equation}\label{XOsde}
    dX(t) = (\nu(t) - \zeta(t,T)) X(t)dt + \sigma(t)X(t)dZ(t).
\end{equation}
Consequently the fair value of a contingent claim is given by the formula
\begin{equation}\label{Ovalformula}
    V(t) = D(t)E_{\mathbb{S}}\left[\left. \frac{1}{D(T)}H \right| \mathcal{F}(t) \right] = D(t)E_{\mathbb{S}}\left[\left. \frac{1}{D(T)}h(X(T)) \right| \mathcal{F}(t) \right].
\end{equation}

From a pragmatic standpoint, what is needed in determining prices of derivatives and managing their risks is to allow sources of uncertainty that are epistemic (or subjective) rather than aleatory in nature. In theory, the value of a derivative can be perfectly replicated by a combination of other derivatives provided that these derivatives are uniquely determined by the formula (\ref{Qvalformula}). In practice, prices of derivatives (such as futures, vanilla options) on the same primary asset are not determined by (\ref{Qvalformula}) from statistically or econometrically observed model (\ref{XPsde}), but made by individual market makers who, with little, if not at all, knowledge of the \emph{true} price process, have used their personal perception of the future. We argue further on this point as follows. If we let $Y = Y(t)$ be the price process of a derivative in a derivative market (such as futures in particular, $Y(t,T) = D(t)E_{\mathbb{S}}\left[\left.X(T)/D(T) \right| \mathcal{F}(t) \right]$, since its contract is not necessarily connected with a physical primary asset), $Y$ must have an abstract dynamics and is assumed to satisfy a SDE
\begin{equation}\label{YQsde}
    dY(t,T) = \nu(t)Y(t,T)dt + \bar{\sigma}(t,T)Y(t,T)dZ(t),
\end{equation}
where $T$ denotes a fixed time horizon larger than or equal to the maturity of any contingent claim, $\bar{\sigma}$ is a Lipschitz continuous function so that a solution exists. We now show that $Z$ is a $\mathbb{S}$-Brownian motion - the source of randomness that drives the derivative price process $Y$. We introduce a change of time, see for example \cite{kle12}. Let $U(t)$ be a positive function such that
\begin{equation*}
    U(t) = \int_0^t \frac{\bar{\sigma}^2(s,T)}{\sigma^2(s)}ds
\end{equation*}
which is finite for finite time $t \le T$ and increases almost surely. Define $\tau(t) = U^{-1}(t)$, let $Y$ be a replacement of $X$, i.e. $X(t) = Y(\tau(t),\tau(T))$ whose solution is given by
\begin{equation*}
    dX(t) = \frac{\nu(t)\sigma^2(t)}{\bar{\sigma}^2(t,T)}X(t)dt + \sigma(t)X(t)dZ(t)
\end{equation*}
with $X(0) = Y(0)$. Rearranging the drift term leads to
\begin{equation}\label{YOsde}
    dX(t) = (\nu(t) - \zeta(t,T))X(t)dt + \sigma(t)X(t)dZ(t),
\end{equation}
where
\begin{equation}\label{zetatimechge}
    \zeta(t,T) = \frac{\nu(t)}{\bar{\sigma}^2(t,T)} \left(\bar{\sigma}^2(t,T) - \sigma^2(t)\right).
\end{equation}
Here, we see the concurrence of the SDEs (\ref{XOsde}) and (\ref{YOsde}), the source of randomness $Z$ is the very $\mathbb{S}$-Brownian motion (\ref{ZObrownian}). We have just shown that the measure $\mathbb{S}$ is subjective in the sense that the valuation of a contingent claim is not only subjected to the dynamics of the primary asset price, but also subject to an exogenous measure of risk $\zeta$. We shall call the measure $\mathbb{S}$ the \emph{risk-subjective measure}. The connection between the risk-subjective measure and the risk-neutral measure described by (\ref{ZObrownian}) is far more precise than that found in \cite{jac00}.

An important note here is that the trading strategy (\ref{Oportmod}) is equivalent to the risk-free money account, that is the growth of portfolio value (\ref{portvalue}) is at the risk-free rate $\nu$. In terms of pricing and hedging, the presence of intrinsic risk imposes an internal consistency and implies that possible arbitrage exists in the market (the primary market and its associated derivative markets).

\section{Applications - pricing and hedging}\label{secapp}
In this section, we shall first discuss some problems related to asset models in parallel markets so as to provide some background for subsequent applications.

In the light of intrinsic risk, the SDE (\ref{YQsde}) in reality may represent a risky asset price process in parallel markets such as: (1) futures price process, or (2) an implied price process recovered from option prices where $\bar{\sigma}$ is the implied volatility. Attempts of recovering the implied price process were pioneered, for examples, by \cite{sch99}, \cite{cfd02}, \cite{let05}, \cite{cw10} and references therein.

Market makers indeed have dispensed with the correct specification (\ref{XPsde}) and directly use an implied price process as a tool to prescribe the dynamics of the implied volatility surface. A practice of recovering an implied price process from observed derivative prices (such as vanilla option prices) and use it to price derivatives is known as \emph{instrumental approach}, described in \cite{reb04}. A practical point that is more pertinent to the instrumental approach is that the prices of exotic derivatives are given by the price dynamics that can take into account or recover the volatility smile. With reference to intrinsic risk, an implied price process is a mis-specification for the primary asset, this was discussed in \cite{ejs98} and was shown that successful hedging depends entirely on the relationship between the mis-specified volatility $\bar{\sigma}$ and the true local volatility $\sigma$, and the total hedging error is given by, assuming zero risk-free rate,
\begin{equation}\label{hedgeerror}
    H - h(X(T)) = \frac{1}{2}\int_0^T X^2(t)\frac{\partial^2V}{\partial x^2} \left(\bar{\sigma}^2(t,T) - \sigma^2(t)\right)dt.
\end{equation}
Note that this hedge error resembles the term (\ref{zetatimechge}). Clearly, the hedging error is an intrinsic price of risk presented as traded asset in the hedging strategy (\ref{Oportmod}), but not in (\ref{Qhedgestrat}).

Before we illustrate a number of applications for pricing and hedging with specific form of the measure of intrinsic risk, let us state a general result for derivative valuation.

\subsection{Risk-subjective valuation}

We have established the risk-subjective valuation formula (\ref{Ovalformula}) where the risk-subjective price process is given by (\ref{XOsde}).
\begin{theorem}
    The risk-subjective value $V$ of a contingent claim $H = h(X(T))$ given by
    \begin{equation}
        V = V(X(t),t) = D(t)E_{\mathbb{S}}\left[\left. \frac{1}{D(T)}h(X(T)) \right| \mathcal{F}(t) \right]
    \end{equation}
    is a unique solution to
    \begin{equation}\label{risksubjectivePDE}
        \frac{\partial V}{\partial t}(x,t) + \frac{1}{2}\sigma^2(t)x^2\frac{\partial^2V}{\partial x^2} + (\nu(t) - \zeta(t,T))x\frac{\partial V}{\partial x}(x,t) = \nu(t) V(x,t),
    \end{equation}
    with $X(t) = x$ and $V(x,T) = h(x)$.
\end{theorem}
\begin{proof}
    The result is obtained by directly applying the Feynman-Kac formula.
\end{proof}

We have shown that the trading strategy (\ref{Oportmod}) yields the risk-free rate of return on the value of a derivative, and also the intrinsic risk is perfectly hedged by delta-hedging representation (\ref{Hrep}).

\subsection{Modelling measure of intrinsic risk}
As unpredictable as a market, prices in a parallel market (such as futures and corresponding vanilla options) may not be driven by the same source of randomness that drives the primary asset (such as stock and bond). Motivated by results (\ref{zetatimechge}) and (\ref{hedgeerror}), in the present framework it makes sense to formulate $\zeta$ by an abstract form
\begin{equation}\label{specZeta}
    \zeta(t,T) = \gamma(t,T)\left(\bar{\sigma}^2(t,T) - \sigma^2(t)\right),
\end{equation}
where $\sigma$ is the volatility of the underlying asset, $\bar{\sigma}$ the volatility of a risky asset in a parallel market.
We propose that $\zeta$ takes a general form of an exponential family
\begin{equation}
    \zeta = e^{\xi(x) + \eta(\theta)\phi(x) - \psi(\theta)},
\end{equation}
the parameter $\theta = \{\sigma, \bar{\sigma}\}$ and $X(t) = x$. As a result, (\ref{specZeta}) is a special case.\\

\textit{Remark.} While the diffusion term $\sigma$ accounts for the distributional property of the primary asset price, the exogenous term $\zeta$ accounts for a phenomenon such as volatility smile. The existence of intrinsic risk appears to undermine the true probability distribution of the underlying, however it emphasises its important role in determining the values of derivatives. It ensures maximal consistency in pricing and hedging contingent claims that are path-dependent/independent and particularly derivatives on volatility (such as variance swap, volatility swap). It insists on a realistic dynamics for the underlying asset as far as delta-hedge is concerned.

\subsection{Valuation of forward and futures contracts}
In practice, forward contracts are necessarily associated with the primary asset (such as stock and bond) and therefore their prices are determined by (\ref{Qvalformula}) and hedged by (\ref{Qhedgestrat}). As was illustrated in the previous section, can be determined by (\ref{Ovalformula}) which includes a measure of intrinsic risk, $\zeta$, as a convexity adjustment.

\subsection{Contingent claims on dividend paying assets with default risk}
Hedgers holding the primary asset in their hedging portfolio would receive dividends which are assumed to be a continuous stream of payments, whereas hedgers holding other hedge instruments (such as futures, vanilla options) do not receive dividends. In this case, $\zeta$ can be considered as dividend yield and $\zeta X dt$ is the amount of dividend received in a time interval $dt$. $\zeta$ may also be a non-negative function representing the hazard rate of default in a time interval $dt$, this well-known approach was proposed in \cite{lin06} and references therein.

\subsection{Foreign market derivatives}
Suppose that $r_f$ is the risk-free rate of return of a foreign money account and $\zeta_v$ the measure of risk that accounts for volatility smile, (\ref{risksubjectivePDE}) is then a direct application to foreign market derivatives where $\zeta = r_f + \zeta_v$. This is indeed the simplest application of risk-subjective valuation.

\subsection{Interest rate derivatives}
As an exogenous variable to the risk-subjective price process (\ref{XOsde}), $\zeta$ of a particular form would become a mean of reversion.  This is a desirable feature in a number of well-known interest rate models such as extended model of \cite{hw90}, \cite{bk91} model.

With reference to the \emph{liquidity preference theory} or the \emph{preferred habitat theory} of \cite{key64}, a term premium for a bond can be represented as a measure of intrinsic risk.

\section{Concluding remarks}
It is well-known among both academic and practitioners that the standard complete market framework often failed, see for example \cite{mp85}. Incomplete market framework becomes crucial in understanding and explaining well-known market anomalies. In this article we have introduced the notion of intrinsic risk and derived the risk-subjective measure $\mathbb{S}$ equivalent to the real-world measure $\mathbb{P}$, where $\mathbb{S}\in\mathcal{Q}$. At a conceptual level, the theory of Girsanov change of measure allows us to recognise that the crucial role of $\mathbb{S}$, rather than the expectation $E_{\mathbb{S}}[H]$, is assigned to the price of a derivative (such as futures, vanilla option). In addition, the intrinsic risk as a structure is what needed to be imposed on the mutual movements of the primary and derivative markets so that, at least, the pricing and hedging derivatives (such as swaps and caplets) can be undertaken on a consistent basis. Apart from such conceptual aspect, the measure $\mathbb{S}$ does not undermine the role of the measure $\mathbb{P}$ in that a lot of knowledge about the primary market is known at any given time $t$. More precisely, the market's expectation (predictions) in terms of a measure $\mathbb{S}$ at time $t$ is given by the conditional probability distribution
\begin{equation}
    \mathbb{S}\left[ \left. \cdot  \right| \mathcal{F}(t)   \right] \mbox{ on } \mathcal{F}_{\mathbb{S}}(t)
\end{equation}
where $\mathcal{F}(t)$ is the information available given by the primary market at time $t$, and $\mathcal{F}_{\mathbb{S}}(t)$ is the information generated by derivatives (such as vanilla options) with maturities $T > t$.

A final remark: In view of the last financial crises, the market has evolved and there is an apparent need, both among practitioners and in academia, to comprehend the problems caused by an excessive dependence on a specific asset modeling approach, by ambiguous specification of risks and/or by confusions between risks and uncertainties (volatilities). As a result, we presented a continuous time framework that, we believe, brings unity, simplicity and consistency to two important aspects: pricing with correctly specified model for a primary risky asset, and hedging risks that can be correctly understood and specified. In addition, the framework proposed in this article is rigorous in the sense that the true meanings of properties and relationships of intrinsic risk and volatility are self-consistent such that their values are not arbitrarily assigned nor should their properties be misused by ignorance.

\end{document}